%% file: PRL_unmarked.tex
\documentclass[twocolumn,superscriptaddress,amsfont,amssymb,amsmath, showpacs,balancelastpage]{revtex4-1}
\usepackage{CJKutf8}

\usepackage{amsthm}
\newtheorem{theorem}{Theorem}

\usepackage[utf8]{inputenc}
\usepackage{amsmath,amssymb,amsfonts,stmaryrd}
\usepackage{physics}
\usepackage{dsfont}
\usepackage{graphicx}
\usepackage{xcolor}
\usepackage{graphicx}
\usepackage{cancel}
\usepackage{natbib}
\usepackage{color}
\usepackage{tikz}
\usepackage{mathrsfs}
\usepackage{relsize}
\usepackage{multirow}
\usepackage[linktocpage]{hyperref}
\usepackage[all]{xy}
\usepackage[export]{adjustbox}
\hypersetup{colorlinks=true,citecolor=blue,linkcolor=blue, urlcolor=blue, breaklinks=true}

\usepackage{float}

\usepackage{bm} 
\usepackage{subfigure} 
\usepackage{environ}
\usepackage{url}
\usepackage[margin=0.75in]{geometry}
\usepackage{ulem}

\usepackage{mathtools}

\newcommand{\Z}{{\mathbb Z}}

\NewEnviron{eqs}{%
\begin{equation}\begin{split}
    \BODY
\end{split}\end{equation}}

\definecolor{dkgreen}{rgb}{0,0.5,0}

\theoremstyle{definition}

\theoremstyle{remark}

\begin{document}

\begin{CJK*}{UTF8}{bsmi}

\title{Exploring Entropic Orders: \\
High Temperature Continuous Symmetry Breaking and Topological Order}

\author{Po-Shen Hsin}
\affiliation{Department of Mathematics, King’s College London, Strand, London WC2R 2LS, UK}
\email[E-mail: ]{po-shen.hsin@kcl.ac.uk}

\author{Ryohei Kobayashi}
\affiliation{Department of Applied Physics, University of Tokyo, Bunkyo-ku, Tokyo 113-8656, Japan}
\email[E-mail: ]{ryohei.k@ap.t.u-tokyo.ac.jp}

\date{\today}
\begin{abstract}
High temperature is usually expected to destroy order: as the Gibbs state approaches the infinite-temperature limit, it becomes an equal-weight ensemble over all states and the system is generically disordered. Recent works showed that entropic order can violate this expectation through coupling to bosons in classical lattice models and quantum field theories, where the ordered states have higher entropy. 
Here we present new analytic methods for constructing quantum lattice models that exhibit entropic orders.
In particular,
we construct quantum lattice models with continuous symmetry breaking at high temperature in 1+1 dimensions and clarify how entropic order can evade the Hohenberg-Mermin-Wagner theorems. We also construct high-temperature entropic $p+ip$ chiral topological superconducting states in 2+1 dimensions with temperature-independent anyon correlation functions. 
In addition, we obtain a broad family of high-temperature entropic non-chiral topological orders. We show that the entropic topological orders have strong higher form symmetries at high temperature unlike the conventional topological orders, and the symmetry is spontaneously broken.
These results follow from two general constructions that couple a given lattice model with a low-temperature ordered phase either to ordered bosons or, for local commuting-projector Hamiltonians, to more general bosonic degrees of freedom.
 \end{abstract}

\maketitle

\end{CJK*}

\textit{Introduction. }
A folklore is that at sufficiently high temperature a quantum system belongs to phase without any order. This is substantiated by the property that the Gibbs state for system with bounded energy will approach the maximally mixed state at sufficiently high temperature $T\rightarrow\infty$, i.e. $\beta=1/T\rightarrow 0$ (we will adopt the units where $k_B=1$):\cite{Kliesch:2014bxj,Bakshi:2024cqr}
\begin{equation}
 (\beta\rightarrow 0)\quad    e^{-\beta H}\rightarrow \mathbf{1}\quad \text{for }\;|H|<{\cal E}~.
\end{equation}
More precisely, \cite{Bakshi:2024cqr} shows that for tensor product Hilbert space with finite local Hilbert space, and the Hamiltonian is a sum of local term $H=\sum H_i$ such that $|H_i|\leq c$ is uniformly bounded by some constant $c$, then there is a critical temperature above which the Gibbs state becomes trivial. This is called ``sudden death of thermal entanglement''.

This property of Gibbs state can be violated if the local Hamiltonian is not bounded. Indeed, it was discovered in recent literature \cite{Han:2025eiw,Huang:2025gvi,Tsao:2026hle,Andriolo:2026udg} that in system with bosons whose particle number can be arbitrarily high $n=0,1,2\cdots$, there can be entropic order -- the Gibbs state does not approach the maximally mixed state, but instead ordered. Similar phenomenon is also discovered in quantum field theory \cite{Chai:2020zgq,Chaudhuri:2020xxb,Bajc:2020gpa,Liendo:2022bmv,Hawashin:2024dpp,Komargodski:2024zmt}.
The ordered states at arbitrarily high temperature are enabled by higher entropy compared to the disordered states, and they are called {\it entropic orders} in \cite{Han:2025eiw}.
 While there exist many constructions for entropic orders including classical statistical lattice models and quantum field theories, more general mechanisms for entropic quantum lattice models are desirable to explore the landscape of entropic orders.

In this Letter, we present two general methods for constructing quantum lattice models that exhibit entropic orders at high temperature from a given quantum lattice model with low-temperature order: a frustration-free Hamiltonian coupled to ordered bosons (Theorem \ref{thm:orderedbos}), and a commuting projector Hamiltonian coupled to general bosons (Theorem \ref{thm:localcommutingprojector}).
Both methods apply to any space dimensions and any lattice.  
We illustrate our method with various explicit models, including high temperature Heisenberg ferromagnetism and high temperature chiral topological orders, as well as high temperature stable entropic topological ordered phases described by local commuting projector models.  
In particular, we show that entropic models enables us to bypass the Hohenberg-Mermin-Wagner theorems, which are no-go theorems for continuous symmetry breaking in spin chain models at low space dimensions \cite{PhysRev.158.383,PhysRevLett.17.1133}. 
The entropic orders discussed here are distinct from related phenomena in which a quantum system becomes ordered at finite, nonzero temperature, such as inverse melting \cite{10.1063-1.1794652,PhysRevE.63.031503}; the Pomeranchuk effect in  
$^3$He, where the solid phase is favored over the liquid phase at higher temperature under suitable pressure \cite{osti_4415012,RevModPhys.69.683}; the order-by-disorder mechanism \cite{Villain1980OrderAA}; and self-correcting quantum memories \cite{Dennis:2001nw,alicki2010thermal,Bombin_2013,Hsin:2024pdi}. In all of these cases, the order disappears once the temperature becomes sufficiently high. This is fundamentally different from the entropic orders which persist to arbitrarily high temperature.

\textit{Entropic Orders From Ordered Bosons. }
Let us consider a local Hamiltonian in the tensor product Hilbert space with finite dimensional local Hilbert spaces, $H_0=\sum H_i$.
We will shift each local Hamiltonian by a constant to make sure that the spectrum of energies is non-negative, $E_I\geq 0$.
The Gibbs state at temperature $T=1/\beta$ is $    \rho_0= \sum_I e^{-\beta E_I} |E_I\rangle\langle E_I|,
$ 
where $|E_I\rangle$ is the energy eigenstate with energy $E_I$.

To simplify the discussion, we will focus on the original model $H_0$ that is frustration free, i.e. $H_0=\sum_i H_i$ such that the ground states minimize each term $H_i$. We will also add a constant to the Hamiltonian such that the energy eigenvalues are non-negative $H_i\geq 0$, hence the ground state has energy $E_0=0$~\footnote{Frustration free model can be gapless or gapped. In the case that the frustration free model is gapless, the dynamical exponent must satisfy $z\geq 2$ \cite{Masaoka:2024oes}. 
Conditions for a local Hamiltonian to be frustration free are discussed in e.g. \cite{Sattath:2016gfh,Masaoka:2025tbj}.}.

In the following, we will couple the above Hamiltonian to two types of bosons with local boson number operators $n_i,m_i\in\{0,1,2,\cdots\}$, such that $n_i,m_i$ mutually commute and also commute with all operators in the original Hilbert space for $H_0$. The resulting model realizes an entropic order, such that a ground state of the original Hamiltonian $H_0$ is realized by a Gibbs state at arbitrary high finite temperature.

The ordered boson entropic models is $H=H_1+H_2$:
\begin{equation}\label{eqn:twobosonsH}
\begin{split}
    H_1 &=\sum_i (a+bH_i)(1+n_i), 
    \\
    H_2 &=\sum_i (a'-b' \delta_{(n_i-n_{i+1}),0})(1+m_i)~,
    \end{split}
\end{equation}
where $a,b,a',b'$ are real constants such that the $a\geq 0$, $b> 0$, and $a'>b'>0$.
The energy spectrum is bounded below since $H_i\geq 0$ by the assumption described above.

\begin{theorem}\label{thm:orderedbos}
Suppose that the Hamiltonian $H_0=\sum_i H_i$ is frustration free, $H_i\ge 0$ and has a unique ground state $\ket{E_0}$ with the energy $E_0=0$.
    Consider the Gibbs state of (\ref{eqn:twobosonsH}) with $a'\rightarrow b'^+,\quad a=0~.$
At any nonzero temperature, the Gibbs state after tracing out the bosons $|n_i\rangle,|m_j\rangle$ produces the ground state of the original model $H_0$.
    
\end{theorem}
See Supplemental Material (SM) \cite{supmat} for the proof. In the proof, the ordered bosons reduce the coupling to the non-commuting local terms of the frustration-free Hamiltonian to a coupling between the full Hamiltonian $H$ and a single collective bosonic occupation number, thereby making the partition function exactly tractable.

When the original model has degenerate ground states $|0\rangle_\alpha$ labeled by $\alpha$, in the entropic model the Gibbs state will produce a mixed state given by equal weight sum of these degenerate ground states $\bigoplus_\alpha|0\rangle_\alpha$ with the limit $\lim_{L\to\infty}\lim_{a\to 0^+}\lim_{a'\to b'^+}$. 
We remark that the entropic model with ordered bosons in theorem \ref{thm:orderedbos} is reminiscent of the entropic order described by continuum field theory with two bosons discussed in \cite{Han:2025eiw}.

\textit{Spontaneously broken continuous symmetry at high temperature. }
We construct quantum lattice models in 1+1d that spontaneously break continuous symmetries at arbitrary high temperature.
The symmetry breaking can be detected by first turning on a small external field $h$ that breaks the symmetry explicitly (and thus lifting the vacua degeneracy), taking the thermodynamic limit, and then taking the limit $h\rightarrow 0$. 

Consider the ferromagnetic Heisenberg model whose ground state spontaneously breaks $SO(3)$ symmetry, coupled to a boson that enables an entropic order at high temperature.
At any fixed finite temperature, for any nonzero $h$ the Gibbs state realizes an exact ground state of the original Heisenberg model breaking the symmetry, therefore the order parameter of the continuous symmetry has nonzero expectation value in the limit $h\rightarrow 0$ in the entropic model at any fixed temperature. Thus we conclude that the continuous symmetry is spontaneously broken in the entropic model at finite temperature.
For spin-$s$, the Heisenberg model is
\begin{align*}
    H_0
    &=J\sum_i \left(2s(2s+1)-(S_i^a+S_{i+1}^a)(S_i^a+S_{i+1}^a)\right)~,
\end{align*}
where $J>0$ is a constant, and $S^\pm_i:=S^x_i\pm iS^y$. The model has $SO(3)$ symmetry, with generator $\tau^{\alpha}=\sum_i S^a_i$ for $\alpha=x,y,z$, as well as time-reversal symmetry. The last line shows that the Hamiltonian is a sum of non-negative terms. Moreover, the ground state minimizes each Hamiltonian term, hence frustration free.

To lift the degeneracy, we add an order parameter
\begin{equation}
    H_0'=H_0+h\sum_i (s+S_i^z)~,
\end{equation}
with $h>0$.  This leaves a unique ground state.

Now, we start with the frustration free model $H_0'$ and couple it to the bosons with number operators $n_i,m_i$. 
Denote $H_0=\sum_i H_i$.
The entropic Heisenberg model is
\begin{equation}
\begin{split}
    H= & \sum_i\left(a+bH_i+bh (s+S_i^z)\right)(1+n_i) \\
    + &\sum_i\left(a'-b'\delta_{(n_i-n_{i+1}),0}\right)(1+m_i)~.
    \end{split}
\end{equation}
For $a\rightarrow 0,a'\rightarrow b'^+$, the Gibbs state at any nonzero temperature for the above 1+1d model with external field $h$ exhibits continuous symmetry breaking: for system size $L$, $\lim_{h\rightarrow 0^+}\lim_{L\rightarrow\infty}
    \lim_{a\rightarrow 0,a'\rightarrow b'^+}\langle S_i^z\rangle=-s\neq 0$.

\textit{How entropic models bypass Hohenberg-Mermin-Wagner theorems.}
The original Hohenberg-Mermin-Wagner (HMW) theorems are no-go theorems for symmetry breaking on certain spin chain models \cite{PhysRev.158.383,PhysRevLett.17.1133} at $T>0$, which is later generalized to the cases $T=0$ in e.g. \cite{10.1143/PTP.54.1039,BSriramShastry_1992} (see  \cite{Halperin_2018} for a review). The theorems are derived using Bogoliubov inequality that originates from Cauchy-Schwarz inequality for correlation functions. Here we will review the arguments following \cite{Watanabe:2023mym} and discuss how the conclusion is modified in the entropic models due to the bosons.

We first review the logic of the original HMW theorem. Consider a 1d spin chain model with the Hamiltonian $H$ with the system size $L$. The Hamiltonian has a spontaneously broken continuous symmetry, whose charge is $Q = \sum_x Q_x$. The order parameter for the spontaneous symmetry breaking (SSB) is given by a local Hermitian operator $O_x$ at the site $x$, with $O:= \sum_x O_x$. We assume the translation invariance of the operators $H, Q, O$.
We define the Hamiltonian with the Zeeman field as $H(h) = H -hO$.
We also define a Hermitian operator $R_i$ at each site $x$, defined through the equation $O_x=[iQ_x,R_x]$. 
The HMW theorem is then derived by the use of Bogoliubov inequality valid at finite temperature $T>0$ \cite{Bogoliubov:1962,PhysRev.158.383,PhysRevLett.17.1133},

\begin{align}
    \frac{1}{L^2} \sum_k\langle R^{\dagger}_k R_k + R_k R^\dagger_k\rangle \ge \frac{1}{L^2} \sum_k \frac{2T |\langle [iQ_k^\dagger, R_k] \rangle|^2}{\langle [Q_k, [H(h), Q_k^\dagger]] \rangle}~,
    \label{eq:MWinequality}
\end{align}
where $\mathcal{O}_k$ is the Fourier transform of the local operators, $\mathcal{O}_k := \sum_x e^{ikx} \mathcal{O}_x$. $\langle \rangle$ denotes the evaluation in the Gibbs ensemble of $H(h)$: $\langle \mathcal{O}\rangle = \mathrm{Tr}(e^{-\beta H(h)}\mathcal{O})/\mathrm{Tr}(e^{-\beta H(h)})$.
Let us also write the denominator of the rhs as
\begin{align}
\frac{1}{L} [Q_k, [H(h), Q_k^\dagger]] = A_k + h B_k~.
\label{eq:Ak+hBk}
\end{align}

The lhs of \eqref{eq:MWinequality} is equal to $\frac{2}{L}\sum_{x}\langle R_x^2 \rangle$, which remains $O(1)$ in the limit of $L\to\infty$. On the rhs of \eqref{eq:MWinequality}, the numerator $\langle [iQ_k^\dagger, R_k] \rangle/L$ approaches the order parameter $\langle O_x\rangle$ in the limit $k\to 0$. Meanwhile, as long as the eigenvalues of local Hamiltonian $H_i$ and local charge $Q_i$ are bounded from above, the denominator of the rhs has vanishing asymptotic behavior $A_k\propto k^2$ as $k\to 0$, as shown in SM \cite{supmat}.

Now, in the limit of $L\to\infty$, the sum over $k$ in \eqref{eq:MWinequality} is converted into the integral over $-\pi\le k\le \pi$. Eq.~\eqref{eq:MWinequality} is then satisfied in the limit $\lim_{h\to 0^+} \lim_{L\to\infty}$ only if
\begin{align}
    \int^\pi_{-\pi}\frac{dk}{2\pi} \frac{2T|m|^2}{A_k} = O(1)~,
\end{align}
where $m:=\langle O_x\rangle_{h\to 0^+}$ is the order parameter. Since $A_k\propto k^2$ as $k\to 0$, one needs to satisfy $m=0$, therefore the symmetry cannot be broken at finite temperature.

Let us now consider the entropic Heisenberg model at $s=1/2$: the Hamiltonian is $H=H_1+H_2$,
\begin{equation}
\begin{split}
    H_1 &=\sum_i \left(a+b(1-X_iX_{i+1}-Y_{i}Y_{i+1} - Z_{i} Z_{i+1})\right)(1+n_i)~, \\
    H_2 &=\sum_i (a'-b' \delta_{(n_i-n_{i+1}),0})(1+m_i)~,
    \end{split}
\end{equation}
where the $U(1)$ symmetry is generated by $Q=\frac{1}{2}X,$ with the order parameter $O=\frac{1}{2} Z, R=-\frac{1}{2}Y$. In this entropic model, \eqref{eq:MWinequality} is compatible with the SSB due to unbounded local Hamiltonian. 
Now, let's consider the limit of $a'\to b'^+, a\to 0^+, L\to \infty$ while fixing $h$ to be finite. After tracing out the bosons $\{m_i,n_i\}$ with this limit, only the polarized ground state of the Heisenberg model satisfying $H_i=Z_i=1$ contributes to the expression,
\begin{align*}
\lim_{a\to 0^+}\lim_{a'\to b'^+} A_k = 2b(1-\cos k)  \lim_{a\to 0^+}\left(\frac{e^{-\beta a}}{1-e^{-\beta a}} \right)\rightarrow\infty~.
\end{align*}
See SM \cite{supmat} for detailed derivations of this divergent behavior.
Hence, the sum over $k$ in the rhs of \eqref{eq:MWinequality} with $k=2\pi j/L, j\in \Z_L$ becomes zero in the limit $\lim_{a\to 0^+}$ due to the diverging denominator, even if the numerator remains nonzero. As a result, \eqref{eq:MWinequality} is trivially satisfied in the entropic model in the limit $    \lim_{h\to 0^+}\lim_{L\to\infty}\lim_{a\to 0^+}\lim_{a'\to b'^+}$,
even if the order parameter in the numerator, $\langle O_x\rangle_{h\to 0^+}$, remains finite. Therefore, the Bogoliubov inequality is compatible with SSB in the entropic model.
A possible field theory description for the entropic model can be obtained from the standard field theory description of the Heisenberg spin chain using compact bosons (see e.g. \cite{Nahum:2025bdg}) and replace the Heisenberg coupling by the boson current for another complex boson that has a quartic potential similar to \cite{Han:2025eiw} so the boson is ordered. We leave the detail analysis of such field theory to future study.

The example also illustrates that ordered bosons are important for the construction.
The goal is to entropically favor the ferromagnetic ground state selected by the field. However, if the bosonic variables independently fluctuate from site to site, they act analogously to a disordered local field or locally varying couplings, which does not coherently favor a uniform ferromagnetic state. The ordered bosonic sector avoids this problem: the bosons are forced into a spatially ordered configuration, so that tracing them out produces an entropic weight depending on the global energy of the original frustration-free Hamiltonian. In the appropriate coupling limit, this selectively enhances the exact ground-state sector of $H_0$.

\textit{Chiral entropic states. }
Let us apply the construction of the entropic model to chiral topological states.
It is known that any local Hamiltonian $H_0=\sum_{i}H_i$ with gapped ground states admits decomposition into frustration-free quasi-local Hamiltonians $H_0=\sum_i \widetilde H_i$, where $\widetilde H_i$ has tails decaying faster than any polynomials \cite{Kitaev_2006}. See also \cite{Sengoku2025} for explicit constructions of frustration free Hamiltonians of gapped free complex fermions with exponentially decaying tails.

For instance, consider a frustration free realization of the $p+ip$ superconductor of free Majorana fermions in 2+1d: $H^{(p+ip)}_0= \sum_i \widetilde H^{(p+ip)}_i~,$
where $\widetilde H^{(p+ip)}_i$ is a Hermitian quadratic operator of Majorana fermions with exponentially decaying tails. 
The entropic model \eqref{eqn:twobosonsH} with this $p+ip$ Hamiltonian $H_0=H^{(p+ip)}_0$ then produces entropic version of chiral topological superconductors.

After gauging the fermion parity symmetry of the above entropic model with $\nu$ copies of $p+ip$ superconductors, the entropic models describe entropic chiral $\mathbb{Z}_2$ topological orders given by $Spin(\nu)_1$ Chern-Simons theory realized at any finite temperature. In particular, for $\nu=2$ this is the Abelian anyon $Spin(2)_1=U(1)_4$, for $\nu=4$ this is the Abelian anyon $Spin(4)_1=SU(2)_1\times SU(2)_1=U(1)_2\times U(1)_2$, for $\nu=8$ this is the three fermion theory $Spin(8)_1=\text{FFF}$.

Consider the correlation functions of the anyon string operators ${\cal W}_i$ such as the vortices in topological superconductors. The operators are independent of the boson operator $n_i$, and thus the correlation functions do not depend on the temperature in the entropic model for $a=0,a'\rightarrow b'^+$:
\begin{equation}
    \text{Tr}\left(e^{-\beta H}\prod_i {\cal W}_i\right)=
    \langle 0|\prod_i {\cal W}_i |0\rangle~,
\end{equation}
where $|0\rangle$ is the ground state of the original model. In particular, when the correlation function contains a string operator ${\cal W}_i$ that creates excitations, then the thermal correlation function vanishes, i.e. no anyons are created from thermal fluctuations.
This is in drastic contrast to the thermal behavior of the $p$-wave topological superconductors (e.g. \cite{R_ising_2019}).

\textit{Exact higher-form symmetries at high temperatures}.
In the above chiral topological order state, since the string operators that create nontrivial excitations have vanishing correlation functions, the Gibbs state has {\it exact (strong) 1-form symmetry} generated by closed loop operators of anyons: ${\cal W}_\text{loop} \rho=\alpha\rho,\quad \rho=e^{-\beta H}$,
where $\alpha$ is a number. For comparison, in the ordinary chiral topological superconductors, the anyon loop symmetry is only a weak (average) symmetry where the symmetry operator commutes with the density matrix, but not a exact (strong) symmetry~\footnote{
For an introduction to the symmetry in mixed states, see e.g. \cite{deGroot:2021vdi}.
}.
Moreover, just as the ground state of the original model has spontaneously broken strong 1-form symmetry down to nothing, since the anyon loops have the same correlation functions in the Gibbs state of the entropic model as the original ground states, the anyon loop strong 1-form symmetry is also spontaneously broken in the Gibbs state of the entropic model. This is true for the Gibbs state of the entropic model at arbitrarily high temperature.

More generally, a finite-temperature Gibbs state on a tensor-product Hilbert space cannot realize an exact strong symmetry whenever the theory admits charged excitations of finite energy. Indeed, at any nonzero temperature, such excitations carry nonzero Boltzmann weight, so the thermal state is no longer confined to the neutral sector required by the strong symmetry. Entropic models evade this obstruction through the divergent entropy of the ground-state sector due to the bosons, which overwhelmingly enhances the statistical weight of strongly symmetric states and correspondingly suppresses the contribution of charged excitations.

We remark that the (approximate) strong 2-form symmetry of finite-temperature Gibbs state is also discussed in e.g. \cite{Zhou:2025bal}. Here, approximate strong symmetry means that the Gibbs state $\rho$ is mapped by a local quantum channel $\mathcal{C}$ into a state with the exact strong symmetry. This approximate strong symmetry is then used to show that the low-temperature thermal state must exhibit long-range entanglement. By contrast, in the entropic models the Gibbs state itself enjoys an \textit{exact} strong higher-form symmetry without the need of a local recovery channel.
This exact strong symmetry persists even at arbitrarily high temperature, rather than only below a finite-temperature threshold.

 \textit{High temperature Hall conductivity in entropic models.}

We start with an integer quantum Hall system of Chern number $\nu$ and construct fractional quantum Hall (FQH) system by gauging the fermion parity. The well-defined charge operator is $Q':=\sum_i 2Q_i$, i.e. the faithful $U(1)$ symmetry is the $\mathbb{Z}_2$ quotient of the previous $U(1)$. The gauged model has Hall conductance $\sigma_{xy}=\frac{e^2}{4h}\nu$.

When the original Chern insulator is realized by a frustration free Hamiltonian, then the resulting bosonic FQH Hamiltonian $H_0$ is also frustration free.
We can couple this system to ordered boson to obtain an entropic model. Let us compute the fractional Hall conductivity of the entropic model in Gibbs state in the limit of couplings where the high temperature Gibbs state reduces to a pure state.
To do so, we can put the system on torus with twisted boundary conditions in the $x,y$ direction by the $U(1)$ transformation $\theta_x,\theta_y$. We denote the twisted FQH Hamiltonian on a torus $H_0(\theta_x,\theta_y)$; this is a frustration free parent Hamiltonian of the FQH state $\ket{\theta_x,\theta_y}$, where the $U(1)$ twist $\theta_x, \theta_y$ is inserted adiabatically starting from the untwisted ground state $\ket{0,0}$. Here, any choice of ground state will give the same Hall conductance. 

The twisted boundary condition can be implemented by the $U(1)$ symmetry defects $V_{\theta_x,\theta_y}$. Consider the following path in the state space: starting from the ground state without twisted boundary condition $(\theta_x,\theta_y)=(0,0)$, we continuously increase the twist to $(\theta_x=2\pi,\theta_y=0)$, and then $(\theta_x=2\pi,\theta_y=2\pi)$, and then decrease it to $(\theta_x=0,\theta_y=2\pi)$, and finally to $(\theta_x=0,\theta_y=0)$. Then the Hall conductance is obtained from the Berry phase of the state along the path \cite{Hastings2014quantization}:
\begin{small}
\begin{equation}
 \langle V_{0,2\pi}^\dag V_{2\pi,0}^\dag V_{0,2\pi}V_{2\pi,0}\rangle=\text{Tr}\left(V_{0,2\pi}^\dag V_{2\pi,0}^\dag V_{0,2\pi}V_{2\pi,0}\rho\right)
=e^{\frac{2\pi i \nu}{4}}~,
\end{equation}
\end{small}
where the Berry phase reduces to the commutator of the $U(1)$ symmetry defects along $x,y$ directions.
In the limit $a'\to b', a=0$, the density matrix becomes the ground state of the original model, and the above equation gives the same fractional Hall conductivity as the original model. 
The Hall conductance discussed above should be understood as a response protected by the strong $U(1)$ symmetry of the entropic Gibbs state. In this setting, it is meaningful to consider adiabatic insertions of $U(1)$ fluxes, or equivalently $U(1)$ twisted boundary conditions, and to define the Hall response by following Laughlin's argument. We note that the system here has strong $U(1)$ symmetry without additional weak symmetry. For more general mixed states in 2+1D that also have additional weak symmetry such as time-reversal, see e.g. \cite{Ma2023average,Ma:2023rji,Ma:2024kma}. The symmetry properties of mixed states can also be treated using symmetry TFT methods as in e.g. \cite{qi2025symmetrytacoequivalencesgapped,  schafernameki2025symtftapproachmixedstates,luo2025topologicalholographymixedstatephases}.
We note that the correlation functions in the Gibbs state is given by averaging over the degenerate ground states in the original model; since every ground state in the original model gives the same Hall conductivity, the average does not do anything.

\textit{Entropic Orders from Local Commuting Projector Models.}
We begin with a local commuting projector model on tensor product Hilbert space:
\begin{equation}\label{eqn:lcp}
    H_0=-\sum P_i,\quad P_i^2=P_i,\quad P_iP_j=P_jP_i~.
\end{equation}

Let us enlarge the Hilbert space by including environment: the new Hilbert space is the original one tensored with an infinite dimensional Hilbert space labeled by arbitrary boson numbers $|n_i\rangle$ for $n_i=0,1,2,\cdots$. We modify the original Hamiltonian $H_0$ to couple the system with the environment by the new Hamiltonian
\begin{equation}\label{eqn:entropicmodelP}
    H=\sum (a-bP_i)(1+n_i)~,
\end{equation}
where $a>b>0$. The boson number operators $n_i$ commute with all operators in the original Hilbert space, $n_i P_j=P_j n_i$.
The ground states satisfy $P_i=1$, $n_i=0$ for all $i$, i.e. the ground states are the same as the original model, with the environment bosons absent.

\begin{theorem}\label{thm:localcommutingprojector}
    The Gibbs state of the entropic model (\ref{eqn:entropicmodelP}) at temperature $T$ with the environment bosons $|n_i\rangle$ traced out is equivalent to the Gibbs state of the clean model (\ref{eqn:lcp}) at temperature $T_\text{eff}$, where $T_\text{eff}$ is related to $T$ by their inverse temperatures $\beta=1/T,\beta_\text{eff}=1/T_\text{eff}$ as
    \begin{equation}\label{eqn:effectivetemp0}
    \beta_\text{eff}=\beta b+\log\left(\frac{1-e^{-\beta a}}{1-e^{-\beta(a-b)}}\right)~.
\end{equation}

\end{theorem}
In particular, even when $\beta$ is finite, $\beta_\text{eff}\rightarrow \infty$ (i.e. approaching zero effective temperature for the original model $H_0$) for $a\rightarrow b^+$.

The proof of Theorem \ref{thm:localcommutingprojector} is found in SM \cite{supmat}.
We remark that a special case of the entropic models for local commuting projectors is discussed in \cite{Han:2025eiw,Tsao:2026hle}, which only consider Calderbank–Shor–Steane (CSS) codes. Note that the entropic versions of these stabilizer codes are not error correcting quantum codes as ordinary stabilizer codes due to zero code distance; instead, one should consider them as subsystem codes, where the bosons correspond to gauge qudits. 
Since we want to engineer the ground state wavefunction of the original model using the Gibbs state of the entropic model, the question is how to prepare the Gibbs states rather than the ground states in a fault-tolerant way such as quantum Gibbs sampling \cite{Kastoryano:2016feb,Rouze:2024ufx,Bergamaschi:2024jls}.

\textit{Entropic non-chiral topological orders and topological phase transitions.}
Using the above construction of entropic models from local commuting projector models, we can construct entropic versions of the Levin-Wen lattice models \cite{Levin:2004mi}, Walker-Wang lattice models \cite{walker201131tqftstopologicalinsulators,von_Keyserlingk_2013},
quantum double models \cite{Kitaev:1997wr,Moradi:2014cfa} and their twisted versions \cite{Hu:2012wx}, fractonic models such as \cite{Haah:2011drr,Slagle:2017mzz,Shirley:2017suz},
as well as lattice models for group cohomology symmetry protected topological (SPT) phases \cite{Chen:2011pg} and beyond group cohomology SPT phases \cite{Fidkowski:2019nju,Chen:2021xks,Fidkowski:2024mof}.

Of particular interest are the thermally stable topological orders such as the 4+1d or 6+1d loop toric codes \cite{Dennis:2001nw,alicki2010thermal,Bombin_2013}. The model exhibits quantum memories at finite temperature $0<T<T_c$.
Since the original loop toric code models have electric and magnetic excitations, different excitations have different critical temperatures depending on the tension of the excitations, called $T_A,T_B$ in \cite{Lu:2019owx}. The critical temperature for quantum memory is $T_c=\text{min}(T_A,T_B)$. For temperature in between $T_c<T<\text{max}(T_{A},T_B):=T_c'$, one type of excitation is suppressed, while the other type of excitation still proliferates, and we have a classical memory with nonzero topological entanglement entropy as shown in \cite{Lu:2019owx}.

In the corresponding entropic model, where the effective inverse temperature is (\ref{eqn:effectivetemp0}) this means that there are following phase transitions from dialing the couplings $a,b$:
(1) the entropic model has stable topological order, (2) the entropic model does not have topological order, but it has classical memory, (3) no topological order and no classical memory. The detailed couplings at these phase separations are given in SM \cite{supmat}.

We remark that the topological phase transitions from dialing the couplings $a,b$ at fixed temperature are described by local commuting projector models. This is in contrast to the conventional transverse field Ising model, where the phase transitions are described by non-commuting Hamiltonians.

\textit{Conclusion.}
In this work, we substantially enlarge the scope of entropic orders by introducing two analytic constructions that realize new forms of entropic order, including continuous symmetry breaking, chiral topological states, and non-chiral stable topologically ordered phases at arbitrarily high temperature.

Several directions remain open. It would be interesting to study their low-free-energy excited states when the continuous symmetry is broken, including possible Goldstone modes in models with ordered bosonic sectors. 
Our constructions mainly realize high-temperature states that reduce, in the appropriate sector, to pure-state orders. It would be valuable to generalize the framework to genuinely mixed-state orders, including mixed-state topological orders and strong-to-weak symmetry-breaking phenomena. Finally, it would be useful to identify physical observables that sharply diagnose these entropic orders, such as thermal and thermoelectric transport coefficients in chiral entropic states. These questions may help clarify how symmetry breaking and topology can persist in thermal quantum matter.

\section*{Acknowledgments}
We thank Haruki Watanabe and Zhu-Xi Luo for helpful discussions.
P-S.H. is supported by Department of Mathematics, King's College London.
R. K. is supported by Department of Applied Physics at University of Tokyo.

\bibliography{biblio.bib}

\onecolumngrid

\vspace{0.3cm}

\newpage

\input{supp_mat.tex}

\vfill



\end{document}

%% file: supp_mat.tex
\begin{center}
\Large{\bf Supplemental Materials}
\end{center}
\onecolumngrid

\section{Proof of Theorem 1}

\noindent {\bf Theorem 1.} 
Suppose that the Hamiltonian $H_0=\sum_i H_i$ is frustration free, $H_i\ge 0$ and has a unique ground state $\ket{E_0}$ with the energy $E_0=0$. 
Suppose the gap scales as $E_1\propto V^{-\eta}$ for some non-negative real number $\eta$ (when $\eta=0$, the system is gapped).
Consider the Gibbs state of (\ref{eqn:twobosonsH}) with
    \begin{equation}
 0\leq a'- b'\ll a',\quad 0\leq a\ll {\cal O}\left(V^{-(1+\eta)}\right)~.
    \end{equation}
\begin{itemize}
    \item[(1)] When $a,(a'-b')$ satisfy the above condition and nonzero, the model $H=H_1+H_2$ has the same ground state degeneracy as the original model.

    \item[(2)] At any nonzero temperature, the Gibbs state after tracing out the bosons $|n_i\rangle,|m_j\rangle$ produces the ground state of the original model $H_0$.
\end{itemize}

\begin{proof}

Since every local Hamiltonian term in $H_2$ commutes with $H_1$, the Gibbs state at temperature $T=1/\beta$ factorizes into
\begin{equation}
    e^{-\beta H}=e^{-\beta H_1}e^{-\beta H_2}~.
\end{equation}
Let us first trace out the environment bosons $|m_i\rangle$:
\begin{align}
    \text{Tr}_{|m_i\rangle} e^{-\beta H}&=
    e^{-\beta H_1} \text{Tr}_{|m_i\rangle} e^{-\beta H_2}
    =e^{-\beta H_1}\prod_i \frac{e^{-\beta\left(a'-b'\delta_{(n_i-n_{i+1}),0}\right)}}{1-e^{-\beta\left(a'-b'\delta_{(n_i-n_{i+1}),0}\right)}}\cr 
    &=e^{-\beta H_1}\prod_i e^{-\beta'(c-\delta_{(n_i-n_{i+1}),0})}~,
\end{align}
where $c=c(a,b,\beta)$ is a constant, and $\beta'=\beta'(a,b,\beta)$ is the effective inverse temperature. They satisfy
\begin{equation}
    e^{-\beta' c}=\frac{e^{-\beta a'}}{1-e^{-\beta a'}},\quad  e^{-\beta'(c-1)}=\frac{e^{-\beta(a'-b')}}{1-e^{-\beta (a'-b')}}~,
\end{equation}
and thus
\begin{equation}
    \beta'=\beta b'+\log\left(\frac{1-e^{-\beta a'}}{1-e^{-\beta (a'-b')}}\right)~.
\end{equation}
We now take $a'\rightarrow b'^+$, which gives $\beta'\rightarrow \infty$, resulting in the projector to the ground state of the Hamiltonian $\sum_i \left(c-\delta_{(n_i-n_{i+1})}\right)$, i.e. the state with ordered $n_i=n_{i+1}=n$.
Thus 
\begin{equation}
    (a'\rightarrow b'^+):\quad \text{Tr}_{|m_i\rangle} e^{-\beta H}
=e^{-\beta H_1}\mathbb{P}(n_i=n)~.
\end{equation}
Further tracing out $|n_i\rangle$ then produces 
\begin{align}
(a'\rightarrow b'^+):\quad     &\text{Tr}_{|n_i\rangle,|m_i\rangle}e^{-\beta H}=\sum_I \frac{e^{-\beta(aV+bE_I)}}{1-e^{\beta(aV+b E_I)}}|E_I\rangle\langle E_I|
\cr &\;=\frac{e^{-\beta aV}}{1-e^{-\beta aV}}|E_0=0\rangle\langle E_0=0|+\sum_{I>0}\frac{e^{-\beta(aV+bE_I)}}{1-e^{\beta(aV+b E_I)}}|E_I\rangle\langle E_I|~,
\end{align}
where $V$ is the spatial volume.
Now, we take $a=0$ for $E_0=0$ the ground state energy of $H_0$. 
The ground state term dominates over the excited state terms by the factor
\begin{equation}
    \frac{1-e^{-\beta (aV+bE_I)}}{1-e^{-\beta (aV)}}~,
\end{equation}
where $E_I>0$. For models with first excited states that satisfy $E_1\propto V^{-\eta}$, it is sufficient to take $a\ll {\cal O}\left(V^{-(1+\eta)}\right)$ for the ratio to approach infinite, i.e. the ground states dominate over the other states.
As a consistency check, we note that for high temperature $\beta\rightarrow 0$, we can approximate the ratio by
\begin{equation}
    \frac{1-e^{-\beta (aV+bE_I)}}{1-e^{-\beta (aV)}}\approx \frac{aV+bE_I}{aV}=1+\frac{bE_I}{aV}~.
\end{equation}
Thus as long as $aV$ is much smaller than $bE_I$, the ground state contribution dominates over the excited states contribution.
This then gives the ground state $|E_0=0\rangle\langle E_0=0|$ of the $H_0$ at zero temperature.
\end{proof}

\section{Proof of Theorem 3.1}

\begin{proof}

The Gibbs state of $H$ at inverse temperature $\beta=1/T$ with the environment bosons traced out is
\begin{align}
\text{Tr}_{\{|n\rangle\}}e^{-\beta H}&=    \sum_{n_i}\langle n_i| e^{-\beta H}|n_i\rangle
    =
    \prod_i \left(e^{-\beta (a-bP_i)}+e^{-2\beta (a-bP_i)}+\cdots\right)\cr 
    &=\prod_i \left(
    e^{-\beta a}\left(1+(e^{\beta b}-1)P_i\right)
+
    e^{-2\beta a}\left(1+(e^{2\beta b}-1)P_i\right)
    +\cdots
\right)\cr 
&=\prod_i \left(
\frac{e^{-\beta a}}{1-e^{-\beta a}}(1-P_i)+\frac{e^{-\beta(a-b)}}{1-e^{-\beta(a-b)}} P_i
\right)\cr 
&=\prod_i e^{A+BP_i}=\prod_i\left(e^A\left(1+(e^{B}-1)P_i\right)\right)~,
\end{align}
where we have used $P_i^2=P_i$, $e^{\alpha P_i}=1+\alpha P_i+\frac{1}{2!}\alpha^2P_i+\cdots=1+(e^{\alpha}-1)P_i$, and in the last line $A,B$ are
\begin{equation}
    e^A=\frac{e^{-\beta a}}{1-e^{-\beta a}},\quad 
    e^{B}=1+e^{-A}\left(\frac{e^{-\beta(a-b)}}{1-e^{-\beta(a-b)}}-\frac{e^{-\beta a}}{1-e^{-\beta a}}\right)=
    \frac{1-e^{-\beta a}}{e^{-\beta a}}\frac{e^{-\beta(a-b)}}{1-e^{-\beta(a-b)}}=e^{\beta b}\frac{1-e^{-\beta a}}{1-e^{-\beta(a-b)}}~.
\end{equation}
This implies that the thermal density matrix of the modified model $H$ is equivalent to that of the original model $H_0$ at the effective inverse temperature $\beta_\text{eff}=B$:
\begin{equation}\label{eqn:effectivetemp}
    \beta_\text{eff}=\beta b+\log\left(\frac{1-e^{-\beta a}}{1-e^{-\beta(a-b)}}\right)~.
\end{equation}
 In other words, the entropic order model $H$ at nonzero temperature $T$ tracing out the environment is equivalent to the ordinary topological order model $H_0$ at the above effective temperature $T_\text{eff}=1/\beta_\text{eff}$.
\end{proof}

\section{Detail of how entropic model evades HMW Theorem}

For spin-$s$, the Heisenberg model is
\begin{align}
    H_0&=J\sum_i \left(s^2-\frac{1}{2}\left(S_i^+S_{i+1}^-+S_i^-S_{i+1}^+\right)-S_i^zS_{i+1}^z\right)\cr
    &=J\sum_i \left(2s(2s+1)-(S_i^a+S_{i+1}^a)(S_i^a+S_{i+1}^a)\right)~.
\end{align}

To see the vanishing denominator on the right hand side of (\ref{eq:MWinequality}) as $k\to 0$, we express the denominator as
\begin{align}
\frac{1}{L} [Q_k, [H(h), Q_k^\dagger]] = A_k + h B_k~,
\label{eq:Ak+hBk}
\end{align}
where 
\begin{align}
\begin{split}
    A_k &=  \frac{-1}{L} \sum_{x,y} \langle [Q_x, [H,Q_y]] \rangle (1-\cos(k(x-y)))~, \\
    B_k &=  \frac{-1}{L} \sum_{x,y} \langle [Q_x, [O,Q_y]] \rangle \cos(k(x-y))~.
\end{split}
\end{align}
If $H_x, Q_x$ are both bounded, $\langle [Q_x, [H,Q_y]] \rangle$ remains $O(1)$, which leads to vanishing asymptotic behavior $A_k\propto k^2$ as $k\to 0$.

Now, in the limit of $L\to\infty$, the sum over $k$ in \eqref{eq:MWinequality} is converted into the integral over $-\pi\le k\le \pi$. \eqref{eq:MWinequality} is then satisfied in the limit $\lim_{h\to 0^+} \lim_{L\to\infty}$ only if
\begin{align}
    \int^\pi_{-\pi}\frac{dk}{2\pi} \frac{2T|m|^2}{A_k} = O(1)~,
\end{align}
where $m:=\langle O_x\rangle_{h\to 0^+}$ is the order parameter. Since $A_k\propto k^2$ as $k\to 0$, one needs to satisfy $m=0$, therefore the symmetry cannot be broken at finite temperature.

Let us now consider the entropic Heisenberg model at $s=1/2$:

\begin{equation}
\begin{split}
    H &=H_1+H_2~,\quad \\
    H_1 &=\sum_i \left(a+b(1-X_iX_{i+1}-Y_{i}Y_{i+1} - Z_{i} Z_{i+1})\right)(1+n_i)~, \\
    H_2 &=\sum_i (a'-b' \delta_{(n_i-n_{i+1}),0})(1+m_i)~,
    \end{split}
\end{equation}
where the $U(1)$ symmetry is generated by $Q=\frac{1}{2}X,$ with the order parameter $O=\frac{1}{2} Z, R=-\frac{1}{2}Y$. In this entropic model, \eqref{eq:MWinequality} is compatible with the SSB due to unbounded local Hamiltonian. In the entropic model, $A_k$ in the expression \eqref{eq:Ak+hBk} is given by
\begin{align}
\begin{split}
    A_k &= \frac{-1}{L} \sum_{x,y} \langle [Q_x, [H,Q_y]] \rangle (1-\cos(k(x-y))) \\
    &= \frac{2b}{L}(1-\cos k)\sum_i \left\langle (Y_iY_{i+1}+Z_{i}Z_{i+1})(1+n_i)\right\rangle~.
    \end{split}
\end{align}
Note that the eigenvalue of the operator $(Y_iY_{i+1}+Z_{i}Z_{i+1})(1+n_i)$ is unbounded due to the boson operator.

Let us take the limit of $a'\to b'^+, a\to 0^+, L\to \infty$ while fixing $h$ to be finite. After tracing out the bosons $\{m_i,n_i\}$ with this limit, only the polarized ground state of the Heisenberg model satisfying $H_i=Z_i=1$ contributes to the expression,
\begin{align}
\begin{split}
     & \lim_{a\to 0^+}\lim_{a'\to b'^+} \frac{1}{L}\mathrm{Tr}_{n_i,m_i}\left(e^{-\beta H}\sum_i(Y_iY_{i+1}+Z_{i}Z_{i+1})(1+n_i)\right) \\
     &\qquad = \lim_{a\to 0^+} \left[\left(\frac{e^{-2\beta a}}{(1-e^{-\beta a})^2} + \frac{e^{-\beta a}}{1-e^{-\beta a}} \right) \ket{0}\bra{0} + ... \right],\\
     & \lim_{a\to 0^+}\lim_{a'\to b'^+}\mathrm{Tr}_{n_i,m_i}e^{-\beta H}
     =\lim_{a\to 0^+} \left[\left(\frac{e^{-\beta a}}{1-e^{-\beta a}} \right) \ket{0}\bra{0} + ... \right],
     \end{split}
\end{align}
where ... denotes the contribution of excited states which is suppressed by the limit $a\to 0$, and $\ket{0}$ is the ground state with $Z_i=1$. Therefore
\begin{align}
\lim_{a\to 0^+}\lim_{a'\to b'^+} A_k = 2b(1-\cos k)  \lim_{a\to 0^+}\left(\frac{e^{-\beta a}}{1-e^{-\beta a}} \right)~,
\end{align}
which diverges to infinity. Hence, the sum over $k$ in the rhs of \eqref{eq:MWinequality} with $k=2\pi j/L, j\in \Z_L$ becomes zero in the limit $\lim_{a\to 0^+}$ due to the diverging denominator, even if the numerator remains nonzero. As a result, \eqref{eq:MWinequality} is trivially satisfied in the entropic model in the limit 
\begin{align}
    \lim_{h\to 0^+}\lim_{L\to\infty}\lim_{a\to 0^+}\lim_{a'\to b'^+}
\end{align}
even if the order parameter in the numerator, $\langle O_x\rangle_{h\to 0^+}$, remains finite. Therefore, the Bogoliubov inequality is compatible with SSB in the entropic model.

More generally, suppose we assume the following:
\begin{itemize}
\item All operators $O_x$, $Q_x$, and $R_x$ have finite operator norms, uniformly bounded from above;
\item The 1d spin chain Hamiltonian $H$ is strictly local, in the sense that it can be written as a sum of local terms, each supported within an interval with the length $\xi$ for some finite locality length $\xi$; and
\item The order parameter $\langle O_x\rangle$ remains finite in the limit $\lim_{h\to 0}\lim_{L\to\infty}$, so that the continuous symmetry is spontaneously broken.
\end{itemize}
Then the quantity $\langle [Q_x,[H,Q_y]]\rangle$ must diverge for at least one pair of points $(x,y)$ with separation bounded by the locality length, namely $|x-y|\le \xi$. Otherwise, the original Hohenberg-Mermin-Wagner theorem would apply, ruling out spontaneous breaking of the continuous symmetry. This in turn implies that the eigenvalues of the local Hamiltonian within the interval of length $\xi$ must be unbounded.

\section{Quantum Phase Transitions in Entropic Loop Toric Code}

The range of couplings that exhibit topological order is the following:
\begin{equation}\label{eqn:effectivetempp}
    \text{Topological Orders (TO)}:\ \frac{1}{T_c}<\beta b+\log\left(\frac{1-e^{-\beta a}}{1-e^{-\beta(a-b)}}\right)
    .
\end{equation}
This implies that entropic topological order is realized with either of the following temperature ranges:
\begin{align}\label{eqn:enquantum}
\text{Entropic Topological Order}:\quad
 &(1)\quad 
\frac{T}{T_c}<b<a
 \cr 
 &(2)\quad 
        0<b<\frac{T}{T_c}, \ b<a<T\log\left(\frac{1-e^{-\beta_c}}{e^{-\beta b}-e^{-\beta_c}}\right),
\end{align}
For $a,b$ in the above range, the entropic model is in stable topological phase without fine tuning.
For the original model being the 4+1d loop toric code, the transition is known to be the 4+1d Ising fixed point \cite{Weinstein:2019nmz}.

The entropic model exhibits classical memory for couplings that satisfy the following condition:
\begin{equation}
\text{Classical Memory}:\quad     \frac{1}{T_c'}<\beta b+\log\left(\frac{1-e^{-\beta a}}{1-e^{-\beta(a-b)}}\right)<\frac{1}{T_c}~.
\end{equation}
Solving the equation gives either of the following two conditions for realizing entropic classical memory:
\begin{align}\label{eqn:enclassical}
   (1) &\;\frac{T}{T_c'}<b<\frac{T}{T_c},\quad 
    b<a,\quad 
    T\log\left(\frac{1-e^{-\beta_c}}{e^{-\beta b}-e^{-\beta_c}}\right)<a~,\text{ or}\cr 
    (2) &\; 0<b<\frac{T}{T_c'},\quad \cr & b<a,\quad   T\log\left(\frac{1-e^{-\beta_c}}{e^{-\beta b}-e^{-\beta_c}}\right)<a<T\log\left(\frac{1-e^{-\beta_c'}}{e^{-\beta b}-e^{-\beta_c'}}\right)~.
\end{align}
For $a,b$ in the above range, the entropic model does not have topological order, but it has classical entanglement with nonzero entanglement entropy. The conditions in (\ref{eqn:enclassical}) and (\ref{eqn:enquantum}) are mutually exclusive.

\section{Absence of Entropic Order with Classical Random Noise
}

In the entropic models we discussed, we can regard the bosons with different states $|n\rangle$ as a quantum environment. Here, let us explore a different environment given by classical random variables. It is generally expected that the classical random noise tends to suppress ordered state \cite{Imry:1975zz,PhysRevLett.37.1364}, which is unlike the quantum bosons discussed in entropic models. Here, we will demonstrate this expectation explicitly using local commuting projector models coupled to random variables.

Consider the local commuting projector Hamiltonian
\begin{equation}
    H_0=-\sum_i P_i~.
\end{equation}
We couple the model to classical random variables $\eta_i$
\begin{equation}
    H[\eta_i]=\sum_i(a-bP_i)\eta_i~,
\end{equation}
where the random fields at different location $i$ are not correlated, and they obey the Gaussian distribution $P(\eta_i)= \frac{1}{\sqrt{2\pi \sigma}}e^{-\eta_i^2/(2\sigma)}$. The coupling constants $a,b$ are some real numbers and we do not make assumption about them other than being real and finite.

The averaged Gibbs state at temperature $T=1/\beta$ is
\begin{align}
    \rho_\text{avg}
    &=\overline{e^{-\beta H}}=\int \prod d\eta_i P(\eta_i) e^{-\beta H[\eta_i]}=
    \overline{e^{-\beta H}}=\int \prod d\eta_i P(\eta_i) e^{-\beta a\eta_i}\left(1+\left(e^{\beta b\eta_i}-1\right)P_i\right)\cr 
    &=\prod_i\left( e^{\frac{1}{2}\beta^2 a^2\sigma}+\left(e^{\frac{1}{2}\beta^2 (a-b)^2\sigma}-e^{\frac{1}{2}\beta^2 a^2\sigma}\right)P_i\right)=\prod_i e^{A+BP_i}~,
\end{align}
where
\begin{equation}
     e^{B}=e^{\frac{1}{2}\beta^2(b^2-2ab)\sigma}~.
\end{equation}
Thus the state is equivalent to the Gibbs state of the original model $H_0$ at the effective temperature 
$B=\beta_\text{eff}$:
\begin{equation}\label{eqn:random}
\beta_\text{eff} =\frac{1}{2}\beta^2(b^2-2ab)\sigma~.
\end{equation}
Here we see that for any finite couplings $a,b$ and finite deviation $\sigma$ of the random variables, the effective temperature $T_\text{eff}=1/\beta_\text{eff}$  cannot be zero at nonzero finite physical temperature $T=1/\beta$.

In fact, for $b=2a$ the state is the maximally mixed state even at zero physical temperature.

\begin{theorem}
    The averaged Gibbs state with Gaussian random noise that has finite deviation is the maximally mixed state at all zero or non zero temperatures $T$ when
\begin{equation}\label{eqn:randomcondition}
   | b-2a|\ll T^2 ~,
\end{equation}
such as the case $b=2a$.
\end{theorem}
\begin{proof}
    This follows from the relation (\ref{eqn:random}): when $|b-2a|/T^2\rightarrow 0$, $\beta_\text{eff}\rightarrow 0$, i.e. the effective temperature is infinite $T_\text{eff}\rightarrow\infty$.
    Thus the averaged Gibbs state is the maximally mixed state.
\end{proof}
We note that any averaged correlation function of operators that do not contain the random variable $\eta_i$ are given by the correlation function of the original model in the Gibbs state with the effective temperature $T_\text{eff}$,
\begin{equation}
    \overline{\text{Tr}\left(\rho[\eta] {\cal O}_1{\cal O}_2\cdots\right)}=
    \text{Tr}\left(\overline{\rho[\eta]}{\cal O}_1{\cal O}_2\cdots\right)=\text{Tr}\left(\rho_\text{avg}{\cal O}_1{\cal O}_2\cdots\right)=\text{Tr}\left(e^{-\beta_\text{eff}H_0}{\cal O}_1{\cal O}_2\cdots\right)~.
\end{equation}

The result is consistent with the general property that random fields make symmetry breaking harder instead of easier, e.g. the Imry-Ma argument \cite{Imry:1975zz}. Indeed, in this case, with the condition (\ref{eqn:randomcondition}) for the couplings $a,b$, the system cannot be ordered even when the physical temperature approaches zero in general dimensions.

{\bf Entropic order for approximately uniform noise distribution. }
We note that in the extreme case that the random noise approximately has uniform distribution $\sigma\rightarrow\infty$, then entropic order is possible, since if $\sigma\gg 1/|b^2-2ba|$, for any finite temperature $T$ the effective temperature approaches zero. Then the state is equivalent to the zero temperature state of the original model, and it can develop an entropic order.